\documentclass[11pt]{article}
\usepackage{a4wide,amsthm,amsmath,amssymb,graphics,tikz,url}
\usepackage[utf8]{inputenc}
\usepackage[
backend=biber,style=alphabetic,maxnames=8,
sorting=ynt,doi=false,isbn=false,url=false,eprint=false
]{biblatex}
\newtheorem{corollary}{Corollary}
\newtheorem{lemma}{Lemma}
\newtheorem{proposition}{Proposition}
\newcommand{\myqed}{}
\title{On price-induced minmax matchings}

\author{Christoph Dürr\thanks{CNRS, Sorbonne Université, LIP6, France. Work partially done while C.D.\ was affiliated with CNRS, University of Chile, CMM. Funded by the French National Research Agency grant ANR-19-CE48-0016, the Center for Mathematical Modeling grant ANID FB210005, BASAL funds for center of excellence from ANID-Chile,  Deutsche Forschungsgemeinschaft grant BR 4744/2-1, the European
Research Council (ERC) grant 639945 (ACCORD), and the ERC CoG grant  772346 (TUgbOAT).} 
\and Mathieu Mari\thanks{University of Warsaw and IDEAS-NCBR, Poland} \and Ulrike Schmidt-Kraepelin\thanks{Centro de Modelamiento Matemático, Universidad de Chile, Chile. Work partially done while affiliated with TU Berlin, Germany.}}

\begin{document}
\maketitle

\usetikzlibrary{decorations.pathmorphing}

\begin{abstract}
We study a natural combinatorial pricing problem for sequentially arriving buyers with equal budgets. Each buyer is interested in exactly one pair of items and purchases this pair if and only if, upon arrival, both items are still available and the sum of the item prices does not exceed the budget. The goal of the seller is to set prices to the items such that the number of transactions is maximized when buyers arrive in adversarial order. 

Formally, we are given an undirected graph where vertices represent items and edges represent buyers. Once prices are set to the vertices, edges with a total price exceeding the buyers' budgets are evicted. 
Any arrival order of the buyers leads to a set of transactions that forms a maximal matching in this subgraph, and an adversarial arrival order results in a \emph{minimum} maximal matching. In order to measure the performance of a pricing strategy, we compare the size of such a matching to the size of a maximum matching in the original graph. It was shown by Correa et al.\ [IPCO 2022] that the best ratio  any pricing strategy can guarantee lies within $[1/2, 2/3]$. Our contribution to the problem is two-fold: First, we provide several characterizations of subgraphs that may result from pricing schemes. Second, building upon these, we show an improved upper bound of $3/5$ and a lower bound of $1/2 + 2/n$, where $n$ is the number of items. 

    
\end{abstract}


\section{Introduction}
Imagine a seller, who has a set $V$ of items and needs to choose prizes
$p_v\geq 0$ for each $v\in V$. 
There are buyers interested in purchasing exactly two items, provided that the total cost does not exceed their budget and that none of the two items is already sold. We focus on buyers with equal budgets and assume wlog that this budget is $1$. 
The seller knows the preferred bundles of the buyers, but not the order in which they arrive.
The goal is to choose the prices so to maximize
the number of transactions in the worst case, that is in the worst arrival order of the buyers. 
In a more general setting, this problem has been studied for the case of random buyer valuations in \cite{feldman_combinatorial_2015,aardal_optimal_2022}.


The problem can be stated formally in terms of a matching problem. Given a graph $G=(V,E)$, the vertices correspond to the items and the edges to the buyers.\footnote{We remark that, in general, two buyers may request the same pair of items implying that $G$ may have parallel edges. However, since all buyers have equal budget, we can assume without loss of generality that $G$ is simple.} Any feasible set of transactions corresponds to a \emph{matching} in the graph, i.e., a set of edges $M\subseteq E$ such that no two edges share a common endpoint. A matching with the largest cardinality is called \emph{maximum}. A \emph{maximal} matching $M$ is such that it cannot be extended by any additional edge and a maximal matching with minimum cardinality is called a \emph{minmax} matching. The minmax ratio of a graph is the ratio between the cardinalities of a minmax matching and a maximum matching.

It is well known that the minmax ratio of any graph is at least $1/2$, as any maximal matching is at least half of the size of a maximum matching. A maximum matching can be computed in polynomial time, but finding a minmax matching is NP-hard \cite{yannakakis1980edge}, and even hard to approximate within any factor better than $1/2$ under the strong unique game conjecture \cite{koenemann_tight_2021}.

For given prices $p$ for vertices of the graph, the set of edges with total price at most 1 over the endpoints is denoted $E_p=\{uv\in E: p_u + p_v\leq1\}$. We denote by $G_p=(V,E_p)$ the graph restricted to these edges. An edge set $S\subseteq E$ is called \emph{keepable} if there exist prices $p$ such that $S=E_p$. The goal is to choose prices such that the size of a minmax matching in $G_p$ is as large as possible.\footnote{In the auction literature, the objective is typically to maximize the (worst case) sum of collected prices. In our setting, this is essentially equivalent to our objective, since we can enforce that all prices are arbitrarily close to $1/2$.} The performance is measured by comparison with the maximum matching in the original graph $G$. The largest ratio achievable is called the \emph{competitive ratio} of the given graph. Our main goal is to determine the smallest competitive ratio among all graphs, also referred to as the competitive ratio of the problem.
While Correa et al.\ \cite{aardal_optimal_2022} showed that the competitive ratio is in $[1/2, 2/3]$, we contribute to this problem as follows: 
\begin{itemize}
  \item In Section \ref{sec:keepable}, we provide several characterizations of keepable edge sets.
  \item In Section \ref{sec:upper}, we improve the upper bound to 3/5.
  \item In Section \ref{sec:P3}, we provide a minor improvement on the lower bound.
  Formally, we show that the competitive ratio is at least $1/2 + 2/n$, where $n$ is the number of nodes (items). 
\end{itemize}




\subsection*{Related Work}

\paragraph{Posted-price Mechanisms} 
Our problem falls into the realm of designing posted-price mechanisms for combinatorial auctions. In general, buyers may be interested in buying any subset of items, often referred to as \emph{bundles}. A posted-price mechanism sets prices for the items and buyers pay the sum of the item prices of their allocated bundle. While posted-price mechanisms usually fail to maximize social welfare, they are inherently transparent. Much of the research focuses on the existence of \emph{Walrasian equilibria}, i.e., prices for the items such that every buyer obtains a utility maximizing bundle \cite{arrow1951extension,debreu1951coefficient}. In contrast, we deviate from such an envy-freeness requirement by assuming that agents arrive sequentially. 

The model that is closest to ours was initiated by Feldman et al.\ \cite{feldman2015combinatorial} and recently revisited by Correa et al.\ \cite{aardal_optimal_2022}. Here, buyers arrive in a worst-case order, but the valuations of the buyers over bundles are drawn randomly. Correa et al.\ \cite{aardal_optimal_2022} prove the existence of prices that guarantee an $1/(d+1)$-approximation of the optimal social welfare obtained by any allocation mechanism (not necessarily posted-price), where $d$ is the size of a largest bundle with positive valuation by any buyer. Our setting is a special case of theirs obtained by setting $d=2$ and considering buyers' valuations that are \emph{single-minded} (i.e., have positive valuation for exactly one bundle) and deterministic.

Another related and well-studied problem is the \emph{graph pricing problem}. Similar to our setting, buyers are interested in purchasing bundles of size two, however, a crucial assumption is that there is infinite supply of the items (aka \emph{digital goods}) and buyers may have different budgets. Balcan and Blum \cite{balcan2006approximation} show that there exists an efficient algorithm providing a $4$-approximation to the optimal pricing scheme. 
Later, Lee \cite{lee2015hardness} proves that the problem cannot be approximated with a constant smaller than $4$, assuming the unique games conjecture. 
Friggstad and Mahboub \cite{friggstad2021graph} study the graph pricing problem with limited supply, but deviate from our assumption that buyers arrive in a worst-case order.

\paragraph{Minmax Matchings}
Finding a maximal matching of minimum size (i.e., a minmax matching) is a notoriously hard problem, as witnessed by a long line of research: Yannakakis and Gavril \cite{yannakakis1980edge} showed that the minmax matching problem is NP-hard, even in planar graphs or bipartite graphs of maximum degree 3. Moreover, the problem remains NP-hard for the case of planar bipartite graphs, planar cubic graphs \cite{horton1993minimum}, and $k$-regular bipartite graphs \cite{demange2014hardness}. In contrast, there exists a polynomial-time algorithm for series-parallel graphs \cite{richey1988minimum}. 

Regarding approximation, Gotthilf et al.\ \cite{gotthilf20092} showed that there exists a $2 - c\frac{\log(n)}{n}$ approximation algorithm. However, this is essentially tight since the problem is hard to approximate with a constant smaller than $2$ in general graphs \cite{dudycz2019tight} (assuming the unique games conjecture), and even in bipartite graphs \cite{koenemann_tight_2021} (assuming any of two strengthenings of the unique games conjecture). 

\section{Keepable sets}
\label{sec:keepable}

In this section we provide several characterizations of keepable edge sets. In other words, we ask for which edge sets $S\subseteq E$ exist prices $p$ such that $uv \in S$ if and only if $p(u)+p(v)\leq 1$.

For a given set of edges $S\subseteq E$, an \emph{alternating walk} is a sequence of vertices $v_0,v_1,\ldots,v_{\ell}$ with $v_{\ell}=v_0$, $\ell$ even and such that $(v_i,v_{i+1}) \in E$ for every $i$ and it belongs to $S$ if and only if $i$ is odd. Note that a walk can contain an edge several times, as in Figure~\ref{fig:3-2}. A standard argument shows that if there is an alternating walk, then there exists an alternating walk which contains every edge at most twice (once for each direction of traversal).\footnote{If an alternating walk traverses an edge twice in the same direction, then we can decompose the walk into two alternating walks of smaller size. Restricting repeatedly to one of the walks eventually shows the claim.} Alternating walks containing each edge at most once, are called \emph{alternating cycles}. Our first characterization of keepable edge sets is based on alternating walks.

\begin{figure}[ht]
    \centering
    \begin{tikzpicture}[main/.style = {}]
        \node[main] (a) at (-1.207,0.5) {$a$};
        \node[main] (f) at (-1.207,-0.5) {$f$};
        \node[main] (b) at (-0.5,0) {$b$};
        \node[main] (c) at (0.5,0) {$c$};
        \node[main] (d) at (1.207,0.5) {$d$};
        \node[main] (e) at (1.207,-0.5) {$e$};
        \draw (a) -- (f);
        \draw (b) -- (c);
        \draw (d) -- (e);
        \draw[dotted]  (a) -- (b) -- (f);
        \draw[dotted]  (d) -- (c) -- (e);
    \end{tikzpicture}
    \caption{The set of edges $\{af,bc,de\}$ is not keepable, because $(a,b,c,d,e,c,b,f,a)$ is an alternating walk with respect to this set. This graph was used in \cite{aardal_optimal_2022} to show that the competitive ratio is at most $2/3$.}
    \label{fig:3-2}
\end{figure}
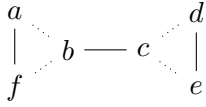

\begin{lemma}   \label{lem:alternating}
    A set $S\subseteq E$ is keepable if and only if there is no alternating walk with respect to~$S$.
\end{lemma}
\begin{proof}
    The proof uses a procedure, which, given $S$, will either generate prices $p$ such that $E_p=S$ or will generate an alternating walk with respect to $S$. For this procedure, let $\epsilon>0$ be a sufficiently small constant, e.g., $1/(2|V|)$ would do. Initialise a counter $i=0$.
    Repeat the following steps while possible.
    \begin{itemize}
        \item If there exists a vertex $u$ incident to only edges in $S$, then set $p(u)=i\cdot \epsilon$; increment $i$ by one; remove $u$ and its incident edges from $G$. 
        \item If there exists a vertex $u$ incident to only edges not in $S$, then set $p(u)=1-i\cdot \epsilon$; increment $i$ by one; remove $u$ and its incident edges from $G$. 
    \end{itemize}
    We consider two cases. If the procedure completed with an empty graph, we claim $E_p=S$. Indeed, a vertex $u$ that received the price $p(u)=i\cdot\epsilon$ was connected at that moment to vertices $v$ which eventually received prices strictly smaller than $1-i\cdot\epsilon$, by the increase of $i$. As a consequence $p(u)+p(v)<1$ as required. On the other hand, a vertex $u$ that received the price $1-i\cdot \epsilon$ was connected at that moment only to vertices $v$, which eventually received prices strictly larger than $i \cdot \epsilon$, and hence $p(u)+p(v)>1$ as required.

    In the case the procedure completed with a non-empty graph $G$, we claim that it contains an alternating walk. Construct a directed graph $H$ from $G$, by making a copy $u'$ of every vertex $u$ in $G$, and generating for every edge $(u,v)$ in $G$ the arcs $(u,v'),(v,u')$ if $(u,v)$ belongs to $S$, and the arcs $(u',v),(v',u)$ if $(u,v)$ does not belong to $S$. This directed graph has positive in- and out-degree. By selecting one incoming and one outgoing arc at every vertex, the selected arcs form a collection of directed cycles. Each of them translates into an alternating walk in $G$. \myqed 
\end{proof}

We remark that the symmetric nature of alternating walks implies that $S$ is keepable if and only its complement $V\setminus S$ is keepable. 

\medskip

The previous proof implies implicitly a different characterization of keepable sets. The procedure described in the proof processes vertices in some order. If $S$ is keepable, 
then every processed vertex $v$ has the following property: Among the not yet processed vertices that are adjacent to $v$, either \emph{all} of them are connected to $v$ via an edge in $S$ or \emph{none} of them is.


\begin{corollary}   \label{cor:order}
    A set $S$ is keepable if and only if there is a total order $\prec$ on the vertices and a subset $R\subseteq V$ such that every edge $uv\in E$ with $u\prec v$ belongs to $S$ if and only if $u\in R$.
\end{corollary}

By grouping vertices along the total order, alternating between vertices in $R$ and outside of $R$, we directly obtain yet another characterization of keepable sets. See Figure~\ref{fig:partition} for illustration.

\begin{corollary}   \label{cor:partition}
    A set $S$ is keepable if and only if there is an ordered partition of the vertices into sets $A_1, A_2, \ldots, A_\ell$, such that an edge $uv$ is in $S$ if and only if $\min\{i,j\}$ is odd, where $i,j$ are such that $u\in A_i, v\in A_j$. The order of the sets is important, hence the name \emph{ordered} partition. For convenience we assume $\ell$ even, and we allow some of the sets to be empty.
\end{corollary}
In some cases we can simplify the ordered partition in such a way that it defines the same keepable set. Suppose that for the partition $A_1, A_2, \ldots, A_\ell$, we have $A_i=\emptyset$ for some $1<i<\ell$. Then replacing the sets $A_{i-1},A_i,A_{i+1}$ by the single set $A_{i-1}\cup A_{i+1}$ preserves the keepable set it defines.

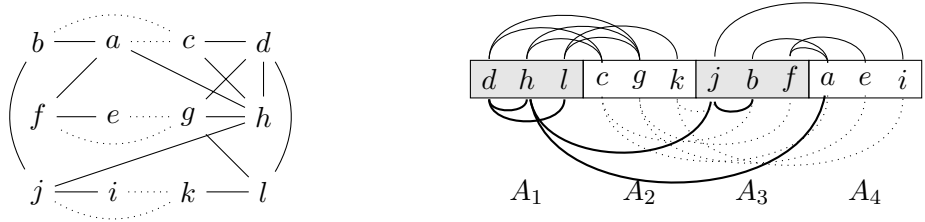
\begin{figure}
    \centering
\begin{tikzpicture}[every circle node/.style={draw}]
\foreach \r/\c/\x [count=\i] in {2/3/a, 1/3/b, 3/3/c, 4/3/d, 2/2/e, 1/2/f, 3/2/g, 4/2/h, 2/1/i, 1/1/j, 3/1/k, 4/1/l}{
	\node (u\i) at (\r,\c) {$\x$};
}
\draw (u1) -- (u2);
\draw (u3) -- (u4);
\draw (u5) -- (u6);
\draw (u1) -- (u6);
\draw (u1) -- (u8) -- (u4) -- (u7) -- (u8) -- (u3);
\draw (u8) -- (u10) -- (u9);
\draw (u11) -- (u12) -- (u7);
\draw[dotted] (u2)  to [bend left] (u3);
\draw (u2)  to [bend right] (u10);
\draw (u12)  to [bend right] (u4);
\draw[dotted] (u1) -- (u3);
\draw[dotted] (u5) -- (u7);
\draw[dotted] (u9) -- (u11);
\draw[dotted] (u7)  to [bend left] (u6);
\draw[dotted] (u11)  to [bend left] (u10);
\draw[fill=gray!20] (6.75,2.25) rectangle  (8.25,2.75) ;
\draw (8.25,2.25) rectangle  (9.75,2.75);
\draw[fill=gray!20] (9.75,2.25) rectangle (11.25,2.75);
\draw (11.25,2.25) rectangle (12.75,2.75);
\foreach \p [count=\i] in {7.5, 9, 10.5, 12}{
    \node at (\p, 1) {$A_\i$};
}
\foreach \r/\x [count=\i] in {11.5/a, 10.5/b, 8.5/c, 7/d, 12/e, 11/f, 9/g, 7.5/h, 12.5/i, 10/j, 9.5/k, 8/l}{
	\node (v\i) at (\r,2.5) {$\x$};
}
\draw (v1) to [in=90,out=90] (v2);
\draw (v3) to [in=90,out=90] (v4);
\draw (v5) to [in=90,out=90] (v6);
\draw (v1) to [in=90,out=90] (v6);
\draw (v4) to [in=90,out=90] (v7);
\draw (v7) to [in=90,out=90] (v8);
\draw (v8) to [in=90,out=90] (v3);
\draw (v10) to [in=90,out=90] (v9);
\draw (v11) to [in=90,out=90] (v12);
\draw (v12) to [in=90,out=90] (v7);
\draw[dotted] (v2) to [in=-90,out=-90] (v3);
\draw[dotted] (v1) to [in=-90,out=-90] (v3);
\draw[dotted] (v5) to [in=-90,out=-90] (v7);
\draw[dotted] (v9) to [in=-90,out=-90] (v11);
\draw[dotted] (v7) to [in=-90,out=-90] (v6);
\draw[dotted] (v11) to [in=-90,out=-90] (v10);
\draw[thick] (v4) to [in=-90,out=-90] (v8);
\draw[thick] (v4) to [in=-90,out=-90] (v12);
\draw[thick] (v2) to [in=-90,out=-90] (v10);
\draw[thick] (v8) to [in=-100,out=-80] (v1);
\draw[thick] (v8) to [in=-100,out=-80] (v10);
\end{tikzpicture}
\caption{The set of solid edges is keepable as testified by the ordered partition on the right, using Corollary~\ref{cor:partition}. Odd indexed sets are depicted in grey.}
    \label{fig:partition}
\end{figure}

The algorithm that proves our lower bound (Section \ref{sec:P3}) 
constructs a set of keepable edges, by starting with all edges and evicting iteratively more and more edges. Such a construction can be done using the previous corollary.

We start by defining the following dynamic labeling of edges. An edge can be labeled \emph{evicted}, \emph{definitely kept} or \emph{temporarily kept}. We define an action, called \emph{refinement} by a given vertex set $B$, which changes the edge labels. Initially all edges are temporarily kept. A refinement by a set $B$ changes the edge labels as follows.
\begin{itemize}
    \item temporarily kept edges $uv$ with both endpoints in $B$ are evicted,
    \item temporarily kept edges $uv$ with no endpoints in $B$ are definitely kept,
    \item the labels of all other edges are unchanged.
\end{itemize}
Note that for given vertex sets, the edge labels depend on the order in which the refinement is done. For example, if an edge has both endpoints in $B$ and none in $C$, then refining by $B$ and then by $C$ will evict the edge. However a refinement by $C$ and then by $B$ will definitely keep the edge.

\begin{figure}[htb]
\begin{center}
\begin{tikzpicture}[every circle node/.style={draw},scale=0.6]
\draw[fill=gray!20,rounded corners=5] (0,1) rectangle  (2,6);
\draw[rounded corners=5] (3,1) rectangle  (5,6);
\draw[fill=gray!20,rounded corners=5] (8,3.5) rectangle  (10,6);
\draw[rounded corners=5] (11,1) rectangle  (13,3.5);
\draw[fill=gray!20,rounded corners=5] (14,1) rectangle  (16,3.5);
\draw[rounded corners=5] (17,3.5) rectangle  (19,6);
\draw[rounded corners=10] (-1,1.5) rectangle (6,3.5);
\draw (5.5, 2.5) node {$B$};
\node (a0) at (1,5) {$a$};
\node (b0) at (4,5) {$b$};
\node (c0) at (1,4) {$c$};
\node (d0) at (4,3) {$d$};
\node (e0) at (1,3) {$e$};
\node (f0) at (4,4) {$f$};
\node (g0) at (1,2) {$g$};
\node (h0) at (4,2) {$h$};
\draw (a0) -- (b0);
\draw (c0) -- (d0);
\draw (e0) -- (f0);
\draw (g0) -- (h0);
\node (a1) at (9,5) {$a$};
\node (b1) at (18,5) {$b$};
\node (c1) at (9,4) {$c$};
\node (d1) at (12,3) {$d$};
\node (e1) at (15,3) {$e$};
\node (f1) at (18,4) {$f$};
\node (g1) at (15,2) {$g$};
\node (h1) at (12,2) {$h$};
\draw[ultra thick] (a1) -- (b1);
\draw (c1) -- (d1);
\draw (e1) -- (f1);
\draw[dotted] (g1) -- (h1);
\node at (1,0.5) {$A_{2i-1}$};
\node at (4,0.5) {$A_{2i}$};
\node at (9,0.5) {$A_{2i-1}\setminus B$};
\node at (12,0.5) {$A_{2i}\cap B$};
\node at (15,0.5) {$A_{2i-1}\cap B$};
\node at (18,0.5) {$A_{2i}\setminus B$};
\end{tikzpicture}
\end{center}
\caption{After a refinement by set $B$, the two sets of the left are replaced by the four sets on the right. Odd indexed sets have gray background. The thick edge is definitively kept, the dotted edge is evicted and all other depicted edges are temporarily kept.}
\label{fig:refinement}
\end{figure}
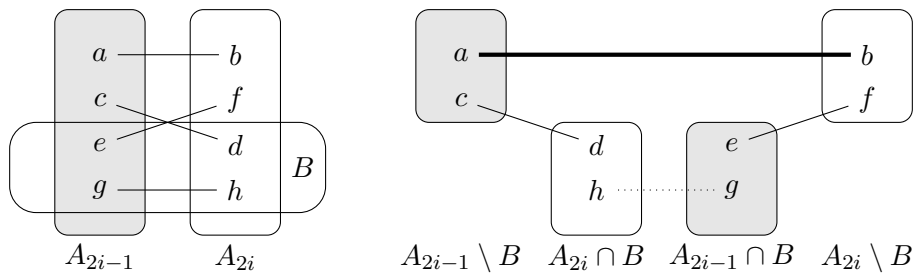

\begin{lemma}\label{lem:refinement}
    An edge set $S$ is keepable if and only if there exists a sequence of vertex sets $B_1,\ldots,B_b$, such that after refinements by these sets, $S$ consists of all temporarily and definitely kept edges.
\end{lemma}
\begin{proof}
We start by proving the backwards direction of the claim and use the keepable edge set characterization of Corollary~\ref{cor:partition}.
We have to show how a refinement by a vertex set $B$ can update the ordered partition, such that the edge labels are as required.
 For a given ordered partition $A_1,\ldots,A_\ell$, we label edges as follows. For an edge $uv\in E$, let $j,k$ be such that $u\in A_j, v\in A_k$. Without loss of generality assume $j\leq k$. If $j$ is even, the edge is \emph{evicted}. If $j$ is odd and $k=j+1$, the edge is \emph{temporarily kept}. But if $j$ is odd and $k=j$ or $k\geq j+2$, the edge is \emph{definitely kept}.  In the right side of Figure~\ref{fig:partition}, dotted edges are evicted, thin solid edges temporarily kept and thick solid edges definitely kept.

    There is one special case to the above defined edge labels. Before the very first refinement, we have the trivial partition $A_1=V$, but we label all edges as temporarily kept.
    
    After the first refinement by set $B$, we have the partition $A_1=V\setminus B, A_2=B$. The edges are labeled as required: Within $A_1$ edges are definitely kept, between $A_1$ and $A_2$ they are temporarily kept and edges within $A_2$ are evicted.

    Now, let $A_1,\ldots,A_\ell$ be an ordered partition obtained by some refinements done in the past, and consider a refinement by a vertex set $B$. We update the partition as follows. Assume $\ell$ is even. For every $1\leq i\leq \ell/2$, we replace the 2 sets $A_{2i-1},A_{2i}$ by the 4 sets $A_{2i-1}\setminus B, A_{2i}\cap B, A_{2i-1}\cap B, A_{2i}\setminus B$ (in this order), see Figure~\ref{fig:refinement}.

    Let $uv\in E$ be an arbitrary edge, and $j,k$ such that $u\in A_j, v\in A_k$. We assume $j\leq k$ without loss of generality. Let $j',k'$ be the index of the sets to which respectively $u,v$ belong after the replacement. 

    Observe that $j$ and $j'$ have the same parity, and so have $k$ and $k'$.

    This already implies that if $uv$ was evicted, then $j$ is even, and so is $j'$, hence $uv$ is still evicted.

    If $uv$ was definitely kept, then $j$ was odd and $k$ was either $j$ or at least $j+2$. In case  $k=j$, then after the replacement we have either $k'=j'$ in case both endpoints or no endpoint belongs to $B$ or $k'=j'+2$ in case exactly one endpoint belongs to $B$. In both cases the edge is still definitely kept. Similarly if $k\geq j+2$, then after the replacement we still have $k'\geq j'+2$, and as a result the edge is still definitely kept.

    If $uv$ was temporarily kept, then $j$ is odd and $k=j+1$. We have to consider 4 cases, depending whether $u$ belongs to $B$ and whether $v$ belongs to $B$. The cases are depicted in Figure~\ref{fig:refinement}.

    In case both endpoints belong to $B$ (edge $gh$ in Figure~\ref{fig:refinement}), we have $j'=k'+1$. Hence before the replacement $u$ belonged to a lower indexed set, and now $v$ belongs to a smaller indexed set. And since this later index is odd, the edge is evicted.

    In case no endpoint belongs to $B$ (edge $ab$ in Figure~\ref{fig:refinement}), we have $k'=j'+3$. As a result the edge is definitely kept.

    In the other two remaining cases (edges $cd$ and $ef$ in Figure~\ref{fig:refinement}), $k'=j'+1$, and the edge remains temporarily kept.

    This establishes that the update of the ordered partition leads to the required edge labeling. Therefore the set of kept edges is keepable.

    It remains to show the other direction of the proof. Given an ordered partition, we show that there is a sequence of vertex sets, which lead by refinements to this partition. It is sufficient to show a single step. Namely, given an ordered partition $C_1,\ldots,C_\ell$ (assuming without loss of generality $\ell$ to be a power of 2), we can construct an ordered partition $A_1,\ldots,A_{\ell/2}$ and a vertex set $B$ such that refining $A$ by $B$ results in $C$. Simply define $B = \bigcup_{i=1}^{\ell/4} C_{4i-2}\cup C_{4i-1}$, and $A_i = C_{2i-1}\cup C_{2i}$. Inspecting the definition of a refinement proves the claim. \myqed
\end{proof}

\section{Upper bound}\label{sec:upper}

In this section, we prove an improved upper bound for the competitive ratio of our problem. To this end, we make use of our characterization of keepable edge sets via alternating walks (Lemma~\ref{lem:alternating}). 

\begin{proposition}
Let $G$ be the Peterson graph. The size of a minmax matching in $G$ is $3$, and there exist no prices $p: E \rightarrow [0,1]$ such that the size of a minmax matching in $G_p$ is larger than $3$.
Thus, the competitive ratio of $G$ is $3/5$. 
\end{proposition}

\begin{proof}
We show that on the Petersen graph, the ratio is 3/5. In this graph, there is an optimum matching of size 5 (see Figure~\ref{fig:petersen-matching}), and we show that the algorithm cannot select a keepable subset of the edges, such that in the induced subgraph the minmax matching has size 4.

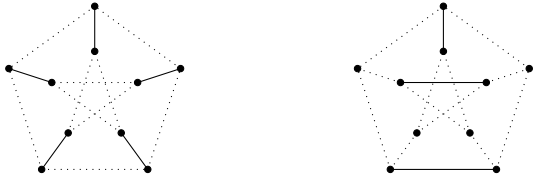
\begin{figure}
    \centering
    \begin{tikzpicture}[scale=0.6]
\draw[fill] (234:2cm) circle (2pt);
\draw[fill] (234:1cm) circle (2pt);
\draw[fill] (306:2cm) circle (2pt);
\draw[fill] (306:1cm) circle (2pt);
\draw[fill] (18:2cm) circle (2pt);
\draw[fill] (18:1cm) circle (2pt);
\draw[fill] (90:2cm) circle (2pt);
\draw[fill] (90:1cm) circle (2pt);
\draw[fill] (162:2cm) circle (2pt);
\draw[fill] (162:1cm) circle (2pt);
\draw[dotted] (234:2cm) -- (306:2cm);
\draw[dotted] (306:2cm) -- (18:2cm);
\draw[dotted] (18:2cm) -- (90:2cm);
\draw[dotted] (90:2cm) -- (162:2cm);
\draw[dotted] (234:2cm) -- (162:2cm);
\draw[dotted] (234:1cm) -- (18:1cm);
\draw[dotted] (306:1cm) -- (90:1cm);
\draw[dotted] (18:1cm) -- (162:1cm);
\draw[dotted] (234:1cm) -- (90:1cm);
\draw[dotted] (306:1cm) -- (162:1cm);
\draw[color=black] (234:2cm) -- (234:1cm);
\draw[color=black] (306:2cm) -- (306:1cm);
\draw[color=black] (18:2cm) -- (18:1cm);
\draw[color=black] (90:2cm) -- (90:1cm);
\draw[color=black] (162:2cm) -- (162:1cm);
\end{tikzpicture}
\hspace{2cm}
\begin{tikzpicture}[scale=0.6]
\draw[fill] (234:2cm) circle (2pt);
\draw[fill] (234:1cm) circle (2pt);
\draw[fill] (306:2cm) circle (2pt);
\draw[fill] (306:1cm) circle (2pt);
\draw[fill] (18:2cm) circle (2pt);
\draw[fill] (18:1cm) circle (2pt);
\draw[fill] (90:2cm) circle (2pt);
\draw[fill] (90:1cm) circle (2pt);
\draw[fill] (162:2cm) circle (2pt);
\draw[fill] (162:1cm) circle (2pt);
\draw[dotted] (306:2cm) -- (18:2cm);
\draw[dotted] (18:2cm) -- (90:2cm);
\draw[dotted] (90:2cm) -- (162:2cm);
\draw[dotted] (234:2cm) -- (162:2cm);
\draw[dotted] (234:2cm) -- (234:1cm);
\draw[dotted] (306:2cm) -- (306:1cm);
\draw[dotted] (18:2cm) -- (18:1cm);
\draw[dotted] (162:2cm) -- (162:1cm);
\draw[dotted] (234:1cm) -- (18:1cm);
\draw[dotted] (306:1cm) -- (90:1cm);
\draw[dotted] (234:1cm) -- (90:1cm);
\draw[dotted] (306:1cm) -- (162:1cm);
\draw[color=black] (90:2cm) -- (90:1cm);
\draw[color=black] (18:1cm) -- (162:1cm);
\draw[color=black] (234:2cm) -- (306:2cm);
\end{tikzpicture}
    \caption{A maximum and a minmax matching in the Petersen graph.}
    \label{fig:petersen-matching}
\end{figure}

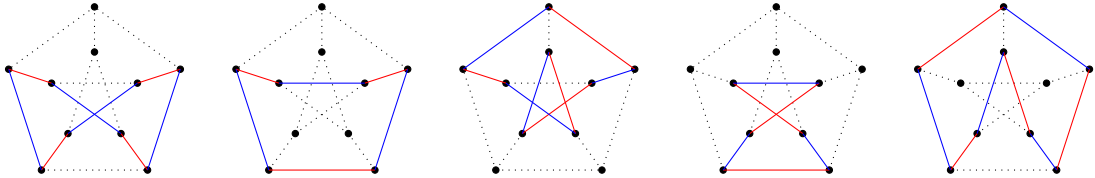
\begin{figure}[ht]
    \centering
\begin{tikzpicture}[scale=0.6]
\draw[fill] (234:2cm) circle (2pt);
\draw[fill] (234:1cm) circle (2pt);
\draw[fill] (306:2cm) circle (2pt);
\draw[fill] (306:1cm) circle (2pt);
\draw[fill] (18:2cm) circle (2pt);
\draw[fill] (18:1cm) circle (2pt);
\draw[fill] (90:2cm) circle (2pt);
\draw[fill] (90:1cm) circle (2pt);
\draw[fill] (162:2cm) circle (2pt);
\draw[fill] (162:1cm) circle (2pt);
\draw[dotted] (234:2cm) -- (306:2cm);
\draw[dotted] (18:2cm) -- (90:2cm);
\draw[dotted] (90:2cm) -- (162:2cm);
\draw[dotted] (90:2cm) -- (90:1cm);
\draw[dotted] (306:1cm) -- (90:1cm);
\draw[dotted] (18:1cm) -- (162:1cm);
\draw[dotted] (234:1cm) -- (90:1cm);
\draw[color=red] (234:2cm) -- (234:1cm);
\draw[color=red] (306:2cm) -- (306:1cm);
\draw[color=red] (18:2cm) -- (18:1cm);
\draw[color=red] (162:2cm) -- (162:1cm);
\draw[color=blue] (306:2cm) -- (18:2cm);
\draw[color=blue] (234:2cm) -- (162:2cm);
\draw[color=blue] (234:1cm) -- (18:1cm);
\draw[color=blue] (306:1cm) -- (162:1cm);
\end{tikzpicture}
\hspace{1em}
\begin{tikzpicture}[scale=0.6]
\draw[fill] (234:2cm) circle (2pt);
\draw[fill] (234:1cm) circle (2pt);
\draw[fill] (306:2cm) circle (2pt);
\draw[fill] (306:1cm) circle (2pt);
\draw[fill] (18:2cm) circle (2pt);
\draw[fill] (18:1cm) circle (2pt);
\draw[fill] (90:2cm) circle (2pt);
\draw[fill] (90:1cm) circle (2pt);
\draw[fill] (162:2cm) circle (2pt);
\draw[fill] (162:1cm) circle (2pt);
\draw[dotted] (18:2cm) -- (90:2cm);
\draw[dotted] (90:2cm) -- (162:2cm);
\draw[dotted] (234:2cm) -- (234:1cm);
\draw[dotted] (306:2cm) -- (306:1cm);
\draw[dotted] (90:2cm) -- (90:1cm);
\draw[dotted] (234:1cm) -- (18:1cm);
\draw[dotted] (306:1cm) -- (90:1cm);
\draw[dotted] (234:1cm) -- (90:1cm);
\draw[dotted] (306:1cm) -- (162:1cm);
\draw[color=red] (234:2cm) -- (306:2cm);
\draw[color=red] (18:2cm) -- (18:1cm);
\draw[color=red] (162:2cm) -- (162:1cm);
\draw[color=blue] (306:2cm) -- (18:2cm);
\draw[color=blue] (234:2cm) -- (162:2cm);
\draw[color=blue] (18:1cm) -- (162:1cm);
\end{tikzpicture}
\hspace{1em}
\begin{tikzpicture}[scale=0.6]
\draw[fill] (234:2cm) circle (2pt);
\draw[fill] (234:1cm) circle (2pt);
\draw[fill] (306:2cm) circle (2pt);
\draw[fill] (306:1cm) circle (2pt);
\draw[fill] (18:2cm) circle (2pt);
\draw[fill] (18:1cm) circle (2pt);
\draw[fill] (90:2cm) circle (2pt);
\draw[fill] (90:1cm) circle (2pt);
\draw[fill] (162:2cm) circle (2pt);
\draw[fill] (162:1cm) circle (2pt);
\draw[dotted] (234:2cm) -- (306:2cm);
\draw[dotted] (306:2cm) -- (18:2cm);
\draw[dotted] (234:2cm) -- (162:2cm);
\draw[dotted] (234:2cm) -- (234:1cm);
\draw[dotted] (306:2cm) -- (306:1cm);
\draw[dotted] (90:2cm) -- (90:1cm);
\draw[dotted] (18:1cm) -- (162:1cm);
\draw[color=red] (18:2cm) -- (90:2cm);
\draw[color=red] (162:2cm) -- (162:1cm);
\draw[color=red] (234:1cm) -- (18:1cm);
\draw[color=red] (306:1cm) -- (90:1cm);
\draw[color=blue] (90:2cm) -- (162:2cm);
\draw[color=blue] (18:2cm) -- (18:1cm);
\draw[color=blue] (234:1cm) -- (90:1cm);
\draw[color=blue] (306:1cm) -- (162:1cm);
\end{tikzpicture}
\hspace{1em}
\begin{tikzpicture}[scale=0.6]
\draw[fill] (234:2cm) circle (2pt);
\draw[fill] (234:1cm) circle (2pt);
\draw[fill] (306:2cm) circle (2pt);
\draw[fill] (306:1cm) circle (2pt);
\draw[fill] (18:2cm) circle (2pt);
\draw[fill] (18:1cm) circle (2pt);
\draw[fill] (90:2cm) circle (2pt);
\draw[fill] (90:1cm) circle (2pt);
\draw[fill] (162:2cm) circle (2pt);
\draw[fill] (162:1cm) circle (2pt);
\draw[dotted] (306:2cm) -- (18:2cm);
\draw[dotted] (18:2cm) -- (90:2cm);
\draw[dotted] (90:2cm) -- (162:2cm);
\draw[dotted] (234:2cm) -- (162:2cm);
\draw[dotted] (18:2cm) -- (18:1cm);
\draw[dotted] (90:2cm) -- (90:1cm);
\draw[dotted] (162:2cm) -- (162:1cm);
\draw[dotted] (306:1cm) -- (90:1cm);
\draw[dotted] (234:1cm) -- (90:1cm);
\draw[color=red] (234:2cm) -- (306:2cm);
\draw[color=red] (234:1cm) -- (18:1cm);
\draw[color=red] (306:1cm) -- (162:1cm);
\draw[color=blue] (234:2cm) -- (234:1cm);
\draw[color=blue] (306:2cm) -- (306:1cm);
\draw[color=blue] (18:1cm) -- (162:1cm);
\end{tikzpicture}
\hspace{1em}
\begin{tikzpicture}[scale=0.6]
\draw[fill] (234:2cm) circle (2pt);
\draw[fill] (234:1cm) circle (2pt);
\draw[fill] (306:2cm) circle (2pt);
\draw[fill] (306:1cm) circle (2pt);
\draw[fill] (18:2cm) circle (2pt);
\draw[fill] (18:1cm) circle (2pt);
\draw[fill] (90:2cm) circle (2pt);
\draw[fill] (90:1cm) circle (2pt);
\draw[fill] (162:2cm) circle (2pt);
\draw[fill] (162:1cm) circle (2pt);
\draw[dotted] (234:2cm) -- (306:2cm);
\draw[dotted] (18:2cm) -- (18:1cm);
\draw[dotted] (90:2cm) -- (90:1cm);
\draw[dotted] (162:2cm) -- (162:1cm);
\draw[dotted] (234:1cm) -- (18:1cm);
\draw[dotted] (18:1cm) -- (162:1cm);
\draw[dotted] (306:1cm) -- (162:1cm);
\draw[color=red] (306:2cm) -- (18:2cm);
\draw[color=red] (90:2cm) -- (162:2cm);
\draw[color=red] (234:2cm) -- (234:1cm);
\draw[color=red] (306:1cm) -- (90:1cm);
\draw[color=blue] (18:2cm) -- (90:2cm);
\draw[color=blue] (234:2cm) -- (162:2cm);
\draw[color=blue] (306:2cm) -- (306:1cm);
\draw[color=blue] (234:1cm) -- (90:1cm);
\end{tikzpicture}\caption{All 5 even length cycles up to rotations.}
    \label{fig:petersen-cycles}
\end{figure}

The proof uses a 01-integer program formulation of the problem, which we run on a solver to certify infeasibility. 

We want to decide if there is an edge set $S$, with the following properties:
\begin{itemize}
    \item $S$ is keepable.
    \item $S$ admits a matching of size $4$.
    \item Every matching of size $3$ in $S$ can be extended to a matching of size $4$.
\end{itemize}
The existence of such a set $S$ would imply that every minmax matching of the graph restricted to edge set $S$ has size at least 4. Hence if no set $S$ has these properties, then the ratio is at most $3/5$ for the Petersen graph.

We introduce a binary variable $x_e$ for every edge $e$, such that $x$ is the characteristic vector of $S$. To ensure that $S$ is keepable, we use the characterization involving alternating walks. There are at least 25 even length cycles in the Petersen graph, which are depicted in Figure~\ref{fig:petersen-cycles} up to rotations. For each cycle denote $A,B$ its edge sets taken in alternation. We need to enforce that $S$ is not alternating along the cycle, so we need to prevent $A\subseteq S$ and $B\cap S=\emptyset$ as well as $B\subseteq S$ and $A\cap S=\emptyset$. Denoting by $k=|A|=|B|$, we write the constraint
\[
    1-k \leq \sum_{e\in A}x_e - \sum_{e\in B} x_e \leq k-1.
\]
To enforce that $S$ admits a matching of size $4$ we introduce variables $y_M$ for every matching $M$ of size $4$. We require $\sum y_M\geq 1$ and add constraints such that $y_M=1$ implies $M\subseteq S$. Formally for every matching $M$ of size $4$ we add the constraint
\[
    \sum_{e\in M} x_e \geq 4 y_M.
\]
Finally we introduce for every matching $M$ of size $3$ a variable $z_M$ with the following interpretation. If $M\subseteq S$ then $z_M=1$. So we add the following constraints
\begin{align*}
    \sum_{e\in M} x_e &\leq 2 + z_M \\
    \sum_{e} x_e &\geq z_M,
\end{align*}
where the second sum ranges over all edges $e\not\in M$ such that $M\cup\{e\}$ is a matching.

We used a standard integer linear programming solver to show that the above described linear program has no integer solution.  The instance is so small (250 variables, 431 constraints) that the solver GLPK gives the answer within a second. The source code can be found at \url{https://www.lip6.fr/Christoph.Durr/pricing/petersen.py}.

Note that it is not necessary to show that the Petersen graph contains only those 25 even length cycles, depicted in figure~\ref{fig:petersen-cycles}. Indeed, if we would have missed some walks, our linear program would be a relaxation. And this does not weaken the fact that the solver certifies infeasibility. \myqed
\end{proof}



\section{Lower bound}
\label{sec:P3}

Let $G=(V,E)$ be the given graph and $M^*$ a maximum matching in $G$. We know that the size of every minmax matching in $G$  is at least $|M^*|/2$. Moreover, there are graphs for which this inequality is tight, for example the path of length 3.  In this section we show that it is always possible to restrict the graph to a particular keepable edge set and obtain a minmax matching of strictly larger size.

\begin{proposition}
For every graph $G=(V,E)$, there exist prices $p:E \rightarrow [0,1]$ such that the size of a minmax matching in $G_p$ is at least $(1/2 + 2/n)|M^*|$, where $M^*$ is a maximum matching in $G$. 
\end{proposition}


There are two weaknesses of this result. First, we could not strengthen it to a constant lower bound. Second, even if the proof is constructive, it needs to compute multiple minmax matchings. Since this is an NP-hard problem, it raises the question if such prices can be computed in polynomial time.

\begin{figure}[t]
    \centering
\begin{tikzpicture}[every circle node/.style={draw}]
\draw[fill=gray!20,rounded corners=10] (2.4,0.5) rectangle (5.1, 6);
\node at (3.75, 5.6) {$B$};
\foreach \r/\c/\x [count=\i] in {2/3/a, 1/3/b, 3/3/c, 4/3/d, 2/2/e, 1/2/f, 3/2/g, 4/2/h, 2/1/i, 1/1/j, 3/1/k, 4/1/l}{
	\node (u\i) at (\r*1.5,\c*1.5) {$\x$};
}
\draw[thick, double] (u1) -- (u2);
\draw[thick, double] (u3) -- (u4);
\draw[thick, double] (u5) -- (u6);
\draw (u11) -- (u7) -- (u1) -- (u6);
\draw (u5) -- (u4) -- (u7);
\draw[thick, double] (u7) -- (u8);
\draw (u8) -- (u3);
\draw[thick, double] (u10) -- (u9);
\draw[thick, double] (u11) -- (u12);
\draw (u12) -- (u7);
\draw (u2)  to [bend left] (u3);
\draw[dotted] (u1) -- (u3);
\draw[dotted] (u5) -- (u7);
\draw[dotted] (u9) -- (u11);
\draw (u7)  to [bend left] (u6);
\draw (u11)  to [bend left] (u10);
\end{tikzpicture}
    \caption{Illustration of the algorithm in Section~\ref{sec:P3}. Matching $M^*$ is represented by double lines, and matching $M$ by dotted lines.}
    \label{fig:refinement_algo}
\end{figure}
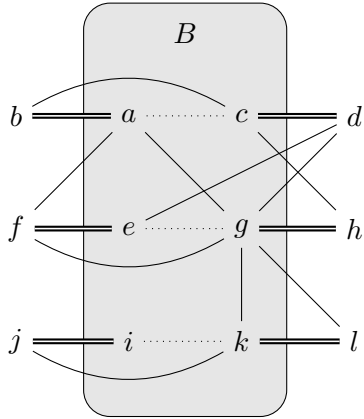

\begin{proof}
Instead of constructing prices directly, we define the following \emph{refinement algorithm}, which makes use of our characterization of keepable edge sets from Lemma~\ref{lem:refinement}.
The algorithm starts by labeling all edges to be temporarily kept and fixes a maximum matching $M^*$. We will later show that, during the entire execution of the algorithm, $M^* \subseteq S$, where $S$ is defined as in Lemma~\ref{lem:refinement}, i.e., as the set of temporarily kept and definitely kept edges. The algorithm repeats the following steps, see Figure~\ref{fig:refinement_algo} for an example. 
\begin{itemize}
    \item Compute a minmax matching $M$ for the graph $H=(V,S)$.
    \item If $M>|M^*|/2$, we are done and stop the procedure.
    \item Otherwise, consider the symmetric difference between $M$ and $M^*$. It consists of a collection of paths of length 3, each consisting of two inner nodes and two endpoints. Let \(B\) be the set of inner nodes of these paths, which are exactly all vertices matched by $M$.
\item Refine the edge labels by $B$, which results in evicting from $S$ all edges with both endpoints in $B$. In particular, all edges in $M$ are evicted.
\end{itemize}
By  Lemma~\ref{lem:refinement}, $S$ is keepable at all times. Moreover, in every iteration, every edge has at least one endpoint in $B$, simply because $M$ is a \emph{maximal} matching. As a result, at any moment, edges are either evicted or temporarily kept, but never definitely kept.

We now argue that $M^* \subseteq S$ during the entire execution of the algorithm. In every iteration, $B$ consists of the inner nodes of length of paths $3$, where edges in $M^*$ form the outer edges. Hence, each edge in $M^*$ has exactly one endpoint in $B$ and therefore remains temporarily kept.

Hence, any minmax matching in $H$ has size at least $|M^*|/2$, and, after termination, has size at least $|M^*|/2 + 1 \geq (1/2 + 2/n) |M^*|$. 

It remains to argue that the algorithm terminates eventually. Since the algorithm evicts a minmax matching of size at least $|M^*|/2$ in each iteration, the procedure terminates after $O(\frac{|V|^2}{|M^*|})$ iterations. \myqed 
\end{proof}

\begin{figure}[htb]
    \centering
\begin{tikzpicture}
\foreach \x/\y/\v in {0/0/G, 2/0/H, 4/0/I, 6/0/J, 8/0/K, 10/0/L,
                      0/10/g, 2/10/h, 4/10/i, 6/10/j, 8/10/k, 10/10/l,
                      2/2/D, 4/2/E, 8/2/F, 2/8/d, 4/8/e, 8/8/f,
                      2/4/A, 4/4/B, 8/4/C, 2/6/a, 4/6/b, 8/6/c}{
	\node (\v) at (\x/2,\y/2) {$\v$};
 }
 \foreach \u/\v in {a/d, A/D, b/e, B/E, c/f, C/F, g/h, G/H, i/j, I/J, k/l, K/L}{
    \draw[thick, double] (\u) -- (\v);
 }
 \foreach \u/\v in {a/A, b/B, c/C}{
    \draw[thick, dotted] (\u) -- (\v);
 }
 \foreach \u/\v in {h/i, H/I, j/k, J/K}{
    \draw[thick, decorate, decoration={snake,amplitude=.4mm,segment length=2mm}] (\u) -- (\v);
 }
 \foreach \u/\v in {d/h, D/H, e/i, E/I, f/k, F/K}{
    \draw[thick] (\u) -- (\v);
 }
\end{tikzpicture}
    \caption{The unique maximum matching $M^*$ is depicted by double lines. In the first iteration, the algorithm receives a minmax matching $M$ of size $7$, depicted by $4$ waved and $3$ dotted lines. The algorithm then refines by the set $\{a,b,c,A,B,C\}$. In the second iteration, the algorithm receives a minmax matching of size $6$, depicted by solid lines, which are all definitely kept by then. The minmax ratio of the resulting graph is $1/2$.}
    \label{fig:refine_too_much}
\end{figure}
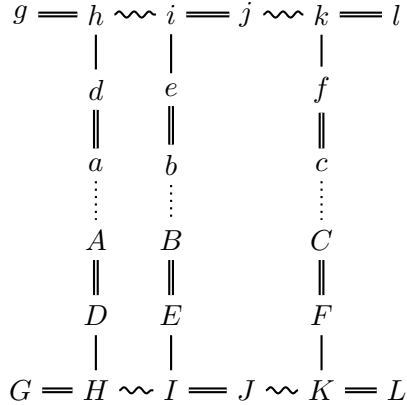

The presented algorithm stops as soon as $M>|M^*|/2$ holds. One might be wondering whether the algorithm could continue evicting edges as long as possible. 
More precisely, the algorithm would still refine by the inner nodes of paths of length three in the symmetric difference $M\Delta M^*$, while it would ignore longer alternating paths. Let's call this the \emph{extended refinement algorithm}.

Unfortunately, this algorithm is not guaranteed to provide an improved minmax ratio. While the refinement strategy would preserve the optimal matching $M^*$, maybe surprisingly, the size of a minmax matching can even decrease to $1/2 |M^*|$.
We discuss such an example in Figure~\ref{fig:refine_too_much}.

\section{Final remarks}

We studied a matching problem which is motivated by posted price mechanisms for online combinatorial auctions. In particular, we were interested in the competitive ratio of our problem, i.e., the ratio between a maximum matching in the original graph, and the maximum size of a minmax matching in any price-induced subgraph. To this end, we gave four characterizations of subgraphs that can be induced by pricing schemes. Building upon these, we were able to reduce the range of the competitive ratio from $[1/2,2/3]$ to $[1/2+2/n,3/5]$.

As the problem is far from being solved, we see many interesting directions for future work. Besides the remaining gap between the lower and upper bound, we remark that the competitive ratio for bipartite graphs could be anywhere between $1/2+2/n$ and $2/3$, as the Peterson graph, which we used for our upper bound, is not bipartite. 



Moreover, our motivation from combinatorial auctions leads to a range of interesting questions in more general settings. For example, can our characterizations be generalized to the case when buyers have different budgets or when buyers are interested in bundles of size larger than $2$? 

While we focused on understanding the \emph{existence} of pricing schemes with good competitive ratio, an orthogonal question is whether an optimal pricing scheme, i.e., one that maximizes the size of induced minmax matchings (per instance), can always be found in polynomial time, and if not, how good can it be approximated? We remark that many related problems in the literature are NP-hard and even hard to approximate.

 


\medskip 

\paragraph{Acknowledgments.} The authors would like to thank José Correa for having introduced this problem to us, as well as for valuable discussions. We would also like to thank Moritz Buchem and Roland Vincze for interesting discussions on the problem. 





\printbibliography 

\end{document}